\documentclass[12pt]{amsart}
\pdfoutput=1
\usepackage[parfill]{parskip}
\usepackage[pdftex]{graphicx,color}
\usepackage{amssymb,amsmath,amsthm}
\usepackage{epstopdf}
\usepackage{wrapfig}
\usepackage{marvosym}
\usepackage{tikz}
\usepackage{verbatim}
\usepackage{caption}
\usepackage{subcaption}

\numberwithin{equation}{section}

\usepackage[pdftex,bookmarks,colorlinks,breaklinks]{hyperref}
\definecolor{dullmagenta}{rgb}{0.4,0,0.4}   
\definecolor{darkblue}{rgb}{0,0,0.4}
\hypersetup{linkcolor=red,citecolor=blue,filecolor=dullmagenta,urlcolor=darkblue}

\usepackage[
textwidth=17.0cm, 
textheight=21cm,
hmarginratio=1:1,
vmarginratio=1:1]{geometry}

\def\dblue#1{\textcolor[rgb]{0,0,0.7}{#1}}

\newcommand{\C}{{\mathbb C}}

\newcommand{\D}{{\mathbb D}}

\renewcommand{\phi}{\varphi}
\renewcommand{\epsilon}{\varepsilon}

\newcommand{\Pol}{\mathrm{Pol}_{n,q}}
\newcommand{\trace}{\mathrm{trace}}
\newcommand{\fluct}{\mathrm{fluct}}

\newtheorem{thm}{Theorem}[section]
\newtheorem{prop}[thm]{Proposition}
\newtheorem{lemma}[thm]{Lemma}

\newtheorem*{thm-others}{Theorem}

\theoremstyle{remark}

\renewcommand{\descriptionlabel}[1]%
         {\dblue{#1:}\\}

\title[Fluctuations in Polyanalytic Ginibre Ensembles]{A Central limit theorem for fluctuations in Polyanalytic Ginibre ensembles} 
\author{Antti Haimi}
\author{Aron Wennman}
\address{Faculty of Mathematics, University of Vienna \newline Oskar-Morgenstern-Platz 1, A-1090 Vienna, Austria}
\address{Department of Mathematics, KTH Royal Institute of Technology,\newline Stockholm, 100 44, Sweden}
\thanks{A.H. was supported by Lise Meitner grant of Austrian Science Fund (FWF)}

\begin{document}

\begin{abstract}
We study fluctuations of linear statistics in Polyanalytic Ginibre ensembles,  a family of point processes describing planar free fermions in a uniform magnetic field 
at higher Landau levels. Our main result is asymptotic normality of fluctuations, extending a result of Rider and Vir\'{a}g.
As in the analytic case, the  variance is composed of independent terms from the bulk and the boundary.  
Our methods rely on a structural formula for polyanalytic polynomial Bergman kernels which separates out the different pure $q$-analytic kernels corresponding to different Landau levels. 
The fluctuations with respect to these pure $q$-analytic Ginibre ensembles are also studied, and a central limit theorem is proved. 
The results suggest a stabilizing effect on the variance when the different Landau levels are combined together.
\end{abstract}
\maketitle

\section{Introduction}
The Ginibre ensemble is one of the major point processes in random matrix theory and mathematical physics. The model has at least three  possible interpretations: in terms of eigenvalues
of random matrices, Coulomb gas of charged particles in an external field, or ground state free fermions in a magnetic field perpendicular to the plane. 
In this paper, we study \emph{Polyanalytic Ginibre ensembles} \cite{haimi2013polyanalytic},  
a family of point process which generalizes the last of these three notions so that the particles are allowed to occupy more general energy levels. 

From the point of view of quantum mechanics, we can arrive at the Polyanalytic Ginibre ensembles in the following way. It is a consequence of the Pauli exclusion principle that 
the wavefunction of the $N$-body system of free (i.e. non-interacting) particles is given by the Slater determinant
\begin{equation} \label{eq:wavefunction}
\det[\psi_{j}(z_j)]_{1 \leq i,j \leq N}
\end{equation}
where $\psi_j$ are orthonormal single particle wave functions. We take these to be eigenstates of the \emph{Landau Hamiltonian} 
$$
H_B:= \frac12 \bigg( \bigg( \frac{\partial}{\partial x} - \frac{B}{2}y \bigg)^2 + \bigg( i \frac{\partial}{\partial y}- \frac{B}{2}x \bigg)^2 \bigg),
$$
which is known to describe a single electron in the plane subjected to a perpendicular magnetic field of strength $B$. It is known  (\cite{eddine1997espaces}, \cite{abreu2015discrete}) that the spectrum of this operator (acting on $L^2(\mathbb{R}^2)$) consists of 
eigenvalues (so called \emph{Landau levels})
$$
e^B_{q}= (q+1/2)B
$$
and the eigenspace corresponding to $e^B_q$ consists of functions of the form 
$$f(z)e^{-B|z|^2/2},$$ 
where $f$ belongs to \emph{pure $q$-analytic Bargmann-Fock space}
$A^2_{\delta; B,q}$, 
defined as the orthogonal difference $A^2_{B,q} \ominus
A^2_{B, q-1}$ between two consequtive \emph{$q$-analytic Bargmann-Fock 
spaces} 
$$
A^2_{B,q}:= \left\{f: \int_{\mathbb{C}} |f(z)|^2 e^{-B|z|^2} dA(z) < \infty,\, \bar \partial^q f(z)=0 \right\}.
$$
The case of the Ginibre ensemble (i.e. $q=1$) corresponds to the classical Bargmann-Fock space, which is associated with the  orthonormal basis 
$$
\psi_j(z)= \frac{B^{1/2}}{\sqrt{j!}} (B^{1/2}z)^je^{-B|z|^2/2}, \quad 0 \leq j.
$$
According to a well-known computation in point process theory, the probability density corresponding to the many-body wavefunction in \eqref{eq:wavefunction} can be written (up to a constant) in a determinantal 
form:
$$
\bigg| \det[\psi_{j}(z_j)]_{1 \leq i,j \leq N} \bigg|^2 \sim \frac{1}{N!} \det[K(z_i,z_j)]_{1 \leq i,j \leq N}e^{-B|z_1|^2- \cdots - B|z_N|^2},
$$
where 
$$
K(z,w)= \sum_{j=0}^{N-1} \phi_j(z) \overline{\phi_j(w)}.
$$
In the Ginibre case, this coincides with the reproducing kernel of the space of analytic polynomials of degree $\leq N-1$ in $L^2(e^{-B|z|^2} dA(z))$.

To obtain the Polyanalytic Ginibre ensembles, we allow wavefunctions from general eigenspaces of $H_B$, not just the from lowest level. The point processes that we will 
consider fall into to two categories: \emph{full type} and \emph{pure type}. In the first case, we have $n$ particles at each level up to $q$ and 
and in the second $n$ particles at $q$'th level only. So, the processes of full type contain $nq$ points, and those of the pure type consist of $n$ points. 
It is natural to take the field  $B$ to be equal to the number of particles $n$ at each level; physically, 
this corresponds to each level being completely filled.  The corresponding reproducing kernels  are formed by choosing the appropriate wavefunctions from $A^2_{B,q}$ and $A^2_{\delta;B,q}$  and 
will be denoted by $K_{n,q}$ and $K_{\delta; n, q}$. They correspond to spaces
$$
\Pol=\mathrm{span} \big\{ \bar{z}^r z^j: 0 \leq j \leq n-1, 0 \leq r \leq q-1 \big\} \subset L^2(e^{-n|z|^2} dA)
$$
and
$$
\delta\Pol: =\Pol\ominus\mathrm{Pol}_{n,q-1},
$$
respectively. We see that when we allow higher Landau levels in the process, the function spaces 
do not only consist of analytic functions as in the Ginibre case, but 
rather more general \emph{polyanalytic functions}. The study of these functions has attracted increasing attention recently, and interesting connections between signal analysis, 
quantum mechanics and complex analysis have been found. For an introduction, see \cite{abreu2014function} or \cite{abreu2015discrete}.


\begin{figure} 
\centering
\begin{subfigure}{.3\textwidth}
  \centering
  \includegraphics[width=\linewidth]{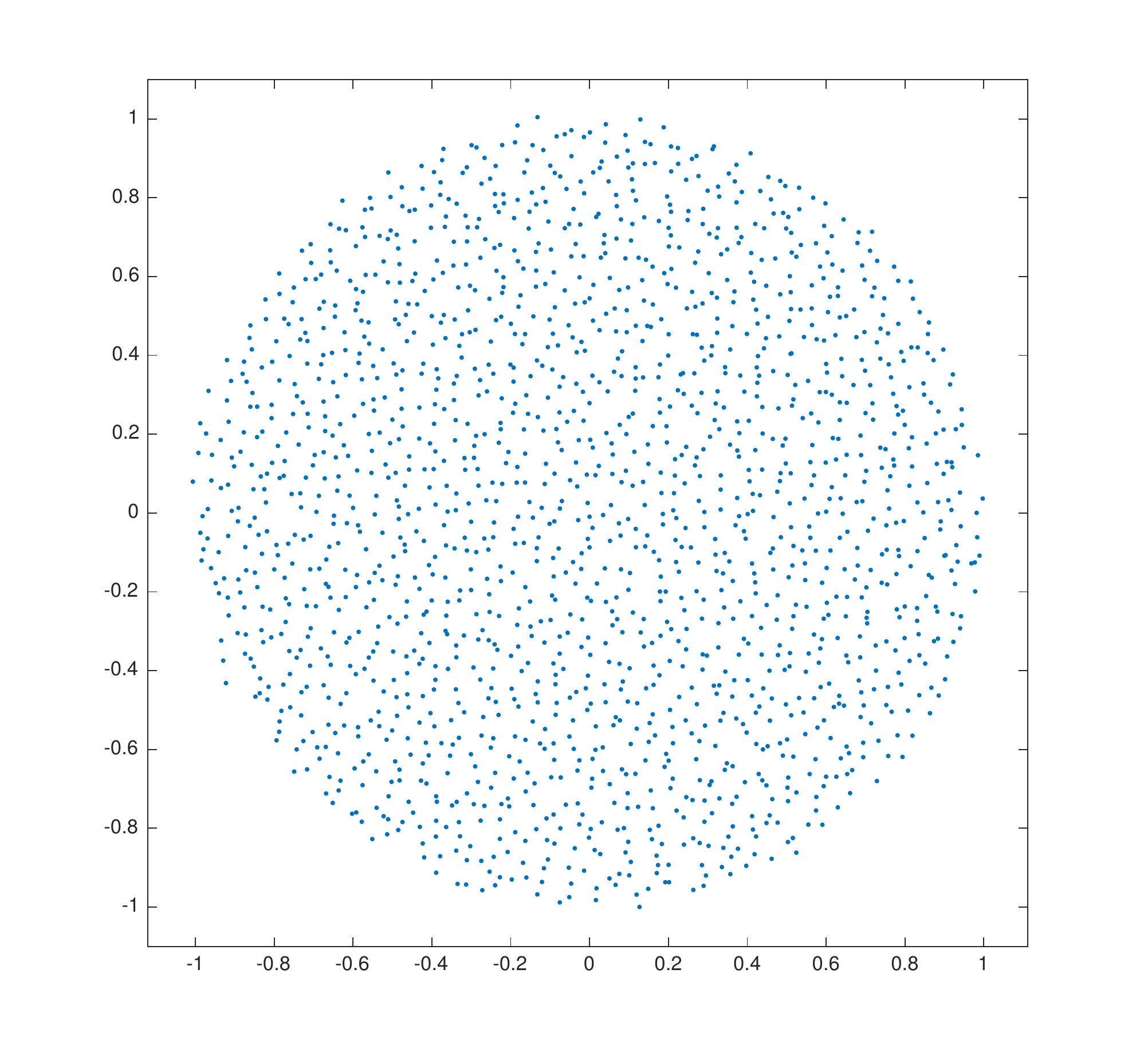}
  \caption{$n=1700$, $q=1$}
  \label{fig:sub1}
\end{subfigure}%
\begin{subfigure}{.3\textwidth}
  \centering
  \includegraphics[width=\linewidth]{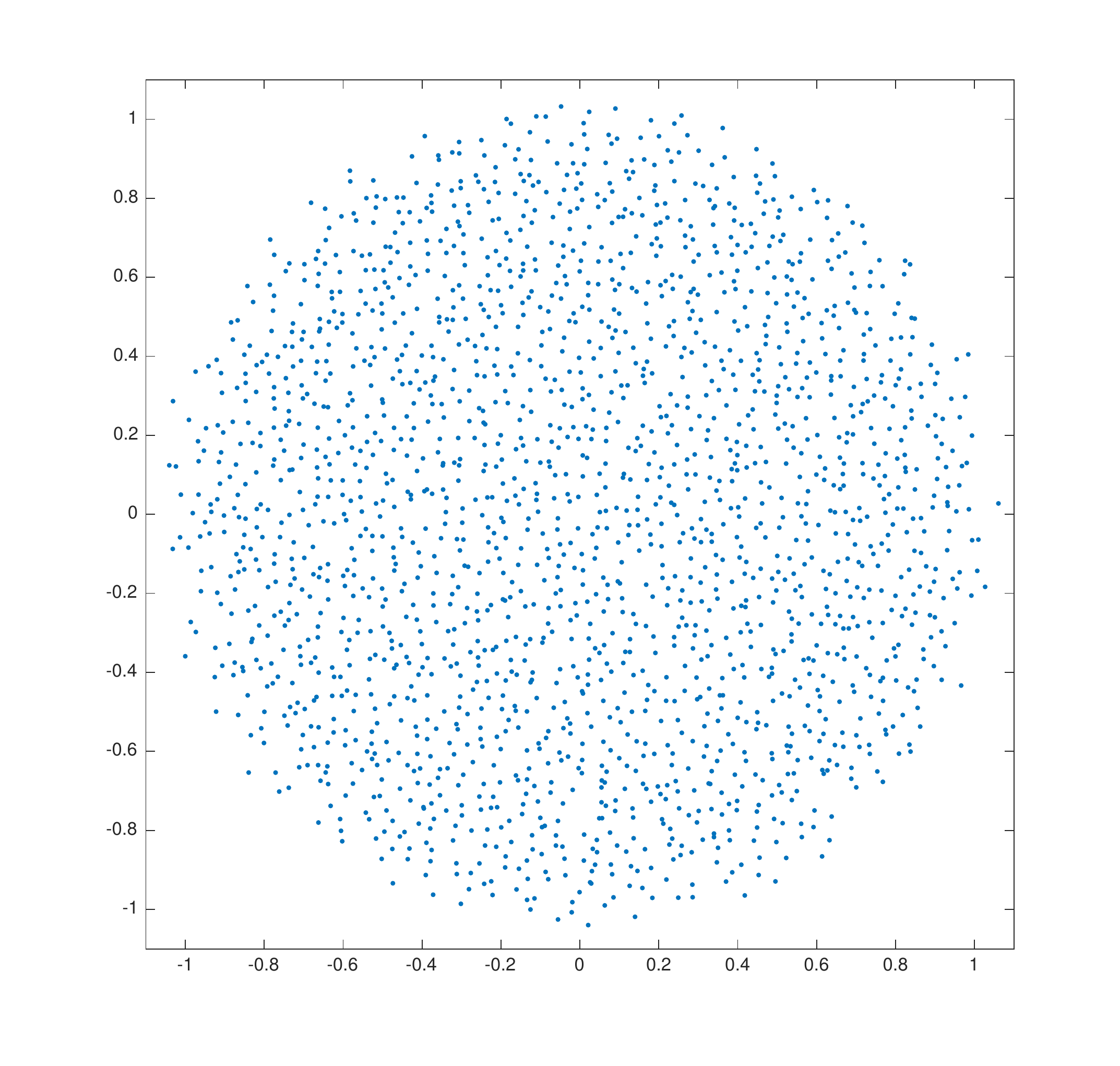}
  \caption{$n=500$, $q=4$}
  \label{fig:sub2}
\end{subfigure}
\begin{subfigure}{.3\textwidth}
  \centering
  \includegraphics[width=\linewidth]{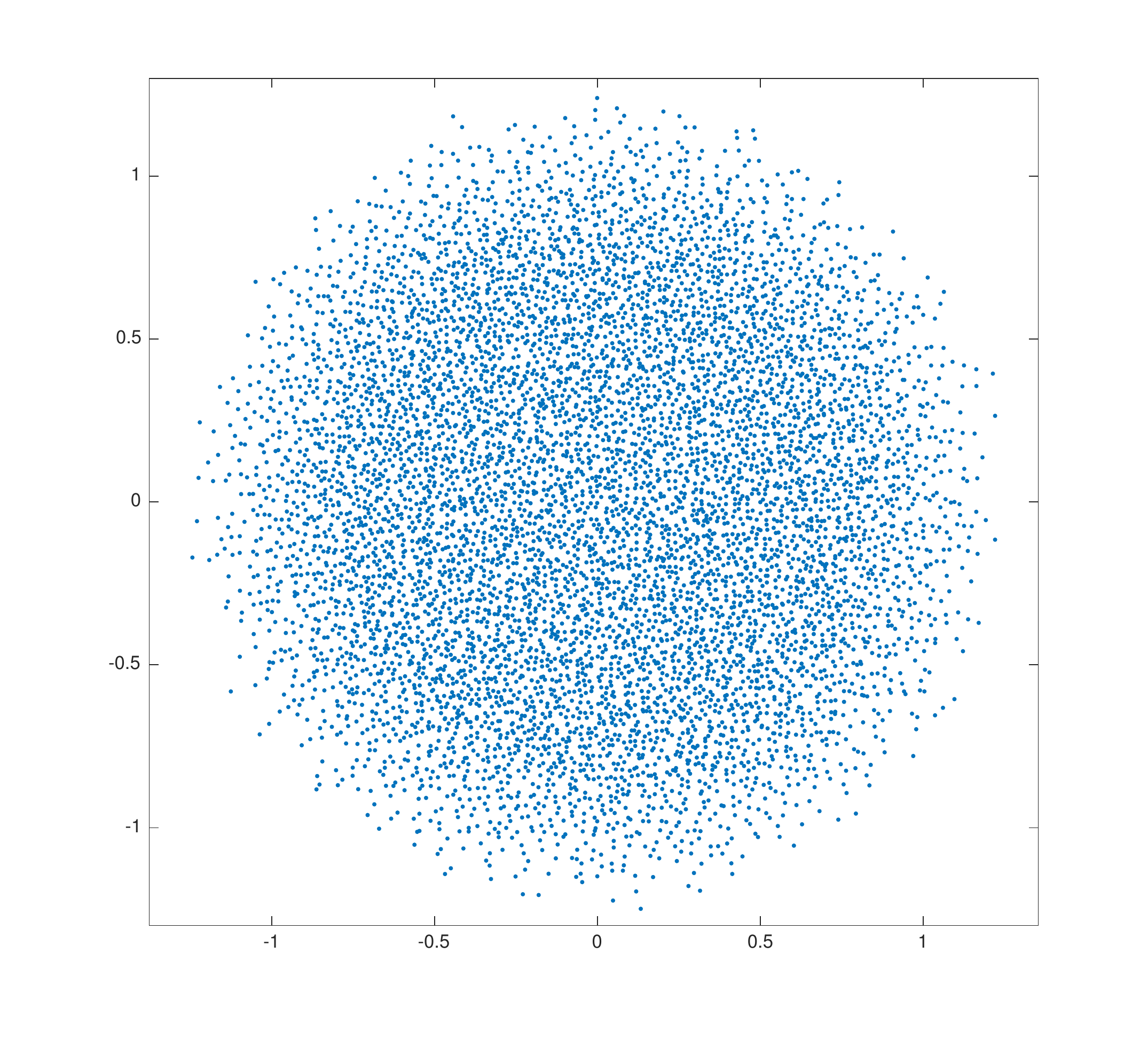}
  \caption{$n=300$, $q=25$}
  \label{fig:sub3}
\end{subfigure}

\caption{Polyanalytic Ginibre ensembles defined by the kernel $K_{n,q}$ with different values of $q$ and $n$. Notice that when $q$ is relatively large compared to $n$, the density is less uniform. }
\label{fig:SimPoly}
\end{figure}

The study of processes of this type was initiated in mathematics literature in \cite{haimi2013polyanalytic}, where precise estimates and scaling limits for the kernels $K_{\delta;n,q}$ and $K_{n,q}$
were obtained when $n \to \infty$ and $q$ is fixed. In particular, it follows from the results there 
that the the circular law will appear in the limit also when higher Landau levels are included (see Figure \ref{fig:SimPoly}). 
Our main theorem shows asymptotic normality of fluctuations around this mean, generalising a result of  Rider and Vir\'{a}g \cite{rider2007noise} from the analytic case. For a continuous test function, we define the linear statistics associated with processes of the pure type 
by
$$
\trace_{\delta;n,q} f = \sum_{j=1}^{n} f(\lambda_j),
$$
where the vector $(\lambda_j)_{j=1}^n$ is picked from the determinantal process defined by the kernel $K_{\delta;n,q}$. The random variables $\trace_{n,q}$ are defined similarly: 
$$
\trace_{n,q} f = \sum_{j=1}^{nq} f(\lambda_j).
$$
Here we have taken $nq$ points from the process defined by $K_{n,q}$. The corresponding fluctuations are defined by
$$
\mathrm{fluct}_{\delta;n,q} = \trace_{\delta; n,q} - \mathbb{E} (\trace_{\delta; n,q}), \qquad \mathrm{fluct}_{n,q}= \mathrm{trace}_{n,q}-\mathbb{E}(\mathrm{trace}_{n,q}).
$$
We denote by $N(a,\sigma^2)$ a normal variable with mean $a$ and variance $\sigma^2$.  

\begin{thm} \label{maintheorem}
Let $g \in C^{\infty}_0(\C)$ be real-valued. We have
\begin{equation}  \label{eq:purevariance}
\mathrm{fluct}_{\delta; q,n}g \to N\big(0, (2q-1) \lVert g \rVert_{H^1(\mathbb{D})}^2 + \frac12 \lVert g \rVert_{H^{1/2}(\partial \mathbb{D})}^2\big),
\end{equation}
and
\begin{equation} \label{eq:fullvariance}
\mathrm{fluct}_{q,n}g \to N\big(0, q ( \lVert g \rVert_{H^1(\mathbb{D})}^2 + \frac12 \lVert g \rVert_{H^{1/2}(\partial \mathbb{D})}^2 \big), 
\end{equation}
in distribution as $n \to \infty$. Here,
$$
\lVert g \rVert_{H^1(\mathbb{D})}^2:= \int_{\mathbb{D}} |\bar \partial g \rvert^2 dA(z)
$$
is the Dirichlet semi-norm of $g$ on $\mathbb{D}$ and 
$$
\lVert g \rVert_{H^{1/2}(\partial \mathbb{D})}^2\:= \sum_{k \in \mathbb{Z}} |k| |\hat{g}(k)|^2, 
$$
is the $H^{1/2}$ semi-norm of $g$ on the unit circle. The numbers 
$$\hat{g}(k):= \frac{1}{2 \pi} \int_0^{2 \pi} g(e^{i \theta}) e^{-ik \theta} d \theta$$ 
are the 
Fourier coefficients of the $g$ restricted to $\{ |z|=1 \}$. 
\end{thm}
We observe that the variance consists of essentially the same two terms as in the case $q=1$, 
the bulk term and the boundary term. 
However, the coefficients in front of the two terms bring in new phenomena. In the formula \eqref{eq:fullvariance}, the variance from the analytic case is just multiplied by $q$. In the pure polyanalytic case \eqref{eq:purevariance}, only the bulk term involves a factor depending on $q$ and the boundary is the same as in the analytic setting. Moreover, 
if we concentrate only on test functions which are supported in the bulk, the variance in 
\eqref{eq:purevariance} is higher than in \eqref{eq:fullvariance}. In fact, because
$\frac1q \sum_{r=1}^q (2r-1)=q$, the variance in \eqref{eq:fullvariance} is obtained by averaging the variances from each of the $q$ Landau levels. 
This suggests that adding several Landau levels together has a certain smoothening effect on the variance. Interestingly, 
for the boundary terms the situation is different, because the boundary contribution in  \eqref{eq:fullvariance} is just the sum of boundary contributions from the individual levels. 
We do not have a physical explanation for these facts at the moment. 

As in \cite{rider2007noise} and \cite{ameur2011fluctuations}, the proof is based on the cumulant method introduced by Costin and Lebowitz
\cite{costin1995gaussian}. Otherwise the argument is quite different. Intead of using explicit expression for the correlation kernel or estimates of it directly, 
our proof starts with a partial integration procedure (Proposition \ref{prop:crossterms}), based on expressing the kernels in terms of quantum mechanical raising operators
which act isometrically between Landau level eigenspaces. As a result, we can rewrite the cumulants in terms of the Ginibre kernel only. This representation combined with 
rather general estimates
of integrals involving cyclic products of kernels allows us to make a reduction to the analytic situation in \cite{rider2007noise}. We want to emphasize that even though precise estimates
and an explicit expressions for our polyanalytic kernels are known, applying formulas for the cumulants directly would most likely lead to be extremely complicated calculations. 

In this paper, the function spaces are defined with a Gaussian weight only. In the analytic case, the result of Rider and Vir\'{a}g has been generalised to more general weights 
by Ameur, Hedenmalm and Makarov (in  \cite{ameur2011fluctuations} for test funtions with support in the bulk, and in \cite{ameur2015random} for general test functions). 
While it is reasonable to expect that results of the former type can be extended to polyanalytic setting by using estimates from \cite{haimi2014bulk}, 
it remains an interesting question what happens with general test functions. The technique in \cite{ameur2015random} seems not to generalize to our polyanalytic case. 
On the other hand, our approach is based on expressing the polyanalytic kernels in terms of the Landau level raising operators and these are closely tied to the Gaussian case. 

Our methods could be used to study other configurations of particles as well, not just those where either one level or all levels up to a given level are filled. However, our techniques do
require the highest level $q$ to be fixed as we let $n \to \infty$. In a forthcoming paper we will study asymptotic behaviour of the kernel $K_{n,q}$ when both $q$ and $n$ tend to infinity. 
Some preliminary observations about this setting can be found in \cite{haimi2013polyanalytic}.

The paper is organised as follows. In Section \ref{sec:prelim}, we provide basic facts about the function spaces and point processes involved.  In Section \ref{sec:cumu}, 
we introduce cumulants and our main technique, the partial integration procedure in Proposition \ref{prop:crossterms}. In Section \ref{sec:cumuest}, we prove certain estimates for 
integrals involving cyclic products of kernels that allow us to estimate the cumulants. The main theorem is then proved in Section \ref{sec:proof}. 

\subsection*{Acknoledgements} 
We would like to thank Joaquim Ortega-Cerd\`{a} for sharing the code which was used to produce the simulations
in Figure \ref{fig:SimPoly} and Seung-Yeop Lee for interesting remarks concerning the main theorem. 

\section{Preliminaries} \label{sec:prelim}

\subsection{Notation}
We write $z= x+iy$ and let $\partial= \frac12 (\partial_x - i \partial_y)$ and $\bar \partial= \frac12 (\partial_x + i \partial_y)$ denote the standard Wirtinger differential operators. 
We will write $\mathbb{D}$ for the open unit disk. We will write $dA(z)= \frac{1}{\pi} dxdy$
and let $d \mu_n(z)=e^{-n|z|^2} dA(z)$ be the Gaussian measure on $\C$. We will use the notation 
$d \mu_n^k(z_1, \ldots, z_k)=d\mu_n(z_1) \ldots d\mu_n(z_k)$. 
\subsection{Spaces of polyanalytic polynomials}
A function $f$ defined on a subset of the complex plane is called $q$-analytic if it satisfies the equation $\bar \partial^q f=0$ in the sense of distributions. Equivalently, a function $f$ is $q$-analytic if it can be decomposed as
$$
f(z)= \sum_{j=0}^{q-1} \bar{z}^j f_j(z),
$$
where the functions $f_j$ are analytic. A function that is $q$-analytic for some $q$ is called \emph{polyanalytic}. We define the \emph{Bargmann-Fock space of $q$-analytic functions} as
\[
A^2_{n,q}:= \big\{ f(z): \int_\C |f(z)|^2 d\mu_n(z), \bar \partial^q f = 0 \big\}.
\]
The spaces of \emph{$q$-analytic polynomials} are defined by
$$
\Pol=\mathrm{span} \big\{\bar{z}^r z^j: 0 \leq j \leq n-1, 0 \leq r \leq q-1 \big\}
$$
We equip $\Pol$ with the inner product from  $A^2(e^{-n\lvert z\rvert^2})$. The \emph{Bargmann-Fock space of pure $q$-analytic functions} is defined as the orthogonal difference
\[
A^2_{\delta;n,q} := A^2_{n,q} \ominus A^2_{n,q-1}. 
\]
The spaces of {\em pure $q$-analytic polynomials} are defined analogously by
$$
\delta\Pol =\Pol\ominus\mathrm{Pol}_{n,q-1}.
$$
Any function space $H$ encountered here possesses a reproducing kernel, i.e. a function $K(z,w)$ such that for any $w\in\C$, $K_w:=K(\cdot,w)$ is an element of the space, and for any function $f\in H$ it holds that
$$
f(w)=\langle f,K_w\rangle,\qquad w\in\C.
$$
We denote by $K_{n,q}(z,w)$ the reproducing kernel for the space $\Pol$ and by $K_{\delta;n,q}(z,w)$ the kernel for the space $\delta\Pol$. The kernel $K_{n,1}(z,w)$ is simply written $K_n(z,w)$.
Clearly, $K_{n,q}$ can be expressed in terms of the pure $q$-analytic kernels as follows:
\begin{equation} \label{eq:kernel-decomp}
K_{n,q}(z,w)=\sum_{r=1}^{q} K_{\delta; n,r}(z,w).
\end{equation}
In \cite{haimi2013polyanalytic}, the following explicit expression for the kernel $K_{n,q}$ 
was given:
\begin{align}
\label{eq:kernelexpli}
K_{n,q} (z,w) = m\sum_{r=0}^{q-1} 
\sum_{i=0}^{n-r-1} \frac{r!}{(r+i)!} 
(nz\bar{w})^i L_r^{i}(n|z|^2)L_r^{i}(n|w|^2) \\
+ n\sum_{j=0}^{q-2}\sum_{k=1}^{q-j-1} \frac{j!}{(k+j)!} 
(n\overline{z}w)^k L^{k}_{j}(n|z|^2)L^{k}_{j}(n|w|^2). \nonumber
\end{align}
The kernel $K_{\delta; n,q}$ can then be obtained from this and the relation
\begin{equation} \label{eq:purekernelexpli}
K_{\delta; n,q}= K_{n,q}(z,w)- 
K_{n, q-1}(z,w).
\end{equation}
 We will not need to use the explicit expressions for these kernels in this paper.  We just observe that 
\[
K_n(z,w)=K_{n,1}(z,w) =
n\sum_{j=0}^{n-1}\frac{(nz \bar w)^j}{j!} \]
is the standard Ginibre kernel from random matrix theory. 
Rather than use \eqref{eq:kernelexpli} and \eqref{eq:purekernelexpli},  we will express 
$K_{n,q}$ and $K_{\delta; n,q}$ in terms of $K_n$ and the \emph{raising operators}
$$
(T_{n,r} f)(z)=\frac{n^{-r/2}}{\sqrt{r!}}e^{n\lvert z\rvert^2}\partial^r\left\{f(z)e^{-n\lvert z\rvert^2}\right\}.
$$
Frequently, it will be convenient to use the notation $T_n := -n^{1/2} T_{n,1}$.
As explained in \cite{haimi2013polyanalytic}, $T_{n,r}$ is an isometric isomorphism from $\mathrm{Pol}_{n,1}$ to $\delta\mathrm{Pol}_{n,r+1}$. 
Thus, $T_{n,r}$ maps orthonormal bases to orthonormal bases. It thus implies that the pure polyanalytic kernels may be obtained from the analytic kernel $K_n(z,w)$ by 
$$K_{\delta;n,r}(z,w)=[T_{n,r-1}]_z\overline{[T_{n,r-1}]}_w K_n(z,w).$$
 Consequently,
\begin{equation}\label{eq:split-kernel}
K_{n,q}(z,w)=\sum_{r=0}^{q-1}[T_{n,r}]_z\overline{[T_{n,r}]}_w K_n(z,w)= \sum_{r=0}^{q-1}\frac{n^{-r}}{r!}[T_n]^r_z \overline{[T_n]}^r_w K_n(z,w).
\end{equation}

Polyanalytic functions appear naturally in e.g. quantum mechanics and time--frequency analysis. The reader who is interested in more background can consult \cite{abreu2014function},  \cite{abreu2010sampling}, \cite{vasilevski2000poly} or \cite{balk1970polyanalytic}. 
\subsection{Polyanalytic Ginibre ensembles} 
Our aim is to study determinantal processes associated with the kernels $K_{n,q}$ and $K_{\delta; n,q}$. These are given by the following probability measures on $\C^{nq}$ and $\C^n$:
\[
d\mathbb{P}_{n,q} := \frac{1}{(nq)!} \det [K_{q,n}(z_j,z_k)]_{1 \leq j,k \leq nq} e^{-n|z_1|^2- \ldots - n|z_{nq}|^2}
dA(z_1) \ldots dA(z_{nq})
\]
and 
\[
d\mathbb{P}_{\delta; n,q} := \frac{1}{n!} \det [K_{\delta; q,n}(z_j,z_k)]_{1 \leq j,k \leq n} e^{-n|z_1|^2- \ldots - n|z_n|^2}
dA(z_1) \ldots dA(z_n).
\]
The fact that these are probability measures is a standard result in theory of determinantal point processes (see e.g. \cite{hough2009zeros} for an introduction). We will identify all copies of $\C$ and interpret $d\mathbb{P}_{n,q}$ and 
$d\mathbb{P}_{\delta; n,q}$ as densities for random configurations of $nq$ and $n$ unlabelled points in $\C$, 
respectively. 
It is well known (\cite{macchi1975coincidence}, \cite{soshnikov2000determinantal}) that any locally trace class projection kernel defines a determinantal process. Because our spaces are finite dimensional, we do not need this general definition here. We just note that an infinite dimensional counterpart of polyanalytic Ginibre ensembles has been studied by Shirai \cite{shirai2015ginibre} and it belongs to 
a more general class of point processes called \emph{Weyl-Heisenberg ensembles}, recently introduced in  
\cite{abreu2016weyl}.



\subsection{Linear statistics}
Recall from the introduction that given a bounded, continuous function $g$,  the linear statistics are defined by
\[
\mathrm{trace}_{n,q}g: = \sum_{j=1}^{nq} g(\lambda_j), 
\]
where $(\lambda_1, \ldots, \lambda_{nq})$ is picked from the measure $d \mathbb{P}_{n,q}$. The variables $\mathrm{trace}_{\delta; n,q}g$ are defined in a similar way, but one picks 
a random vector $(\lambda_1, \ldots, \lambda_{n})$ from $d \mathbb{P}_{\delta; n,q}$ instead. 

Asymptotic behaviour of the expectation of linear statistics can be analysed as follows: 
\[
\frac{1}{nq} \mathbb{E} \mathrm{trace}_{q,n}(f)= \mathbb{E} f(z_1) \\
= \frac{1}{nq} \int_{\mathbb{C}} f(z) K_{q,n}(z,z) e^{-n|z|^2} dA(z) \to
\int_{\mathbb{D}} f(z) dA.
\]
Here, the second equality is a general fact about determinantal point processes. 
The limit follows from the weak convergence $ \frac{1}{nq} K_{q,n}(z,z) e^{-n|z|^2} \to 1_{\overline{\mathbb{D}}}$
which follows from the results in \cite{haimi2013polyanalytic}. A similar statement is true for pure polyanalytic kernels $K_{\delta; n,q}$. This generalizes the well-known circular law
about the Ginibre ensemble (the case $q=1$). The interpretation is that the particles tend to accumulate uniformly on the unit disk. We use standard terminology and refer to the open unit disk as the \emph{bulk}. 

The goal of this paper is to understand how the linear statistics fluctuate around the mean. We let 
$\mathrm{fluct}_{n,q}$ and $\mathrm{fluct}_{\delta; n,q}$ be the mean-zero variables 
\[
\fluct_{n,q}g= \sum_{\lambda_j} g(\lambda_j)- \mathbb{E}g(\lambda_j).
\]
where $\lambda_1, \ldots, \lambda_{nq}$ is picked from the density $d \mathbb{P}_{n,q}$. 
The variables $\fluct_{\delta; n,q}g$ have an analogous definition. 

\section{Cumulants} \label{sec:cumu}
For any real-valued random variable $X$, the cumulants $C_k(X)$ are defined implicitly by
$$
\log \mathbb{E}\left[e^{tX}\right]=\sum_{j=1}^\infty \frac{t^k}{k!} C_k(X).
$$
The first cumulant is equal to the expectation and the second to the variance of $X$. For linear statistics of determinantal processes, the cumulants are given by an explicit formula  
introduced by Costin and Lebowitz in \cite{costin1995gaussian}.  We will use a formulation from \cite{ameur2011fluctuations}. For $X=\trace_{n,q}(g)$, we have
\begin{equation} \label{eq:cumulantformula}
C_k(X)=\int_{\C^k}G_k(z_1,\ldots, z_k) K_{n,q}(z_1,z_2) K_{n,q}(z_2,z_3)\cdots K_{n,q}(z_k,z_1) d\mu_n^k(z_1,\ldots z_k),
\end{equation}
where
$$
G_k(z_1,\ldots z_k)=\sum_{j=1}^k\frac{(-1)^{j-1}}{j}\sum_{\substack{k_1+k_2+\ldots k_j=k \\ k_1,\ldots k_j\geq 1}}\frac{k!}{k_1!k_2!\cdots k_j!}\prod_{l=1}^jg(z_l)^{k_l}.
$$
For $X=\trace_{\delta;n,q}(g)$, one just replaces each occurrence of $K_{n,q}$ by $K_{\delta;n,q}$ in \eqref{eq:cumulantformula}. For future reference, note that $G_k$ is a sum of products of functions of one variable. 

It is known that convergence of the cumulants implies convergence in distribution. Moreover, a random variable is normally distributed iff all cumulants of order $k \geq 3$ vanish. Therefore, to prove the main theorem, we need to show that for $k \geq 3$, $C_k(\trace_{n,q})$ and $C_k(\trace_{\delta; n,q})$ tend to zero and compute the limits of the second cumulants as $n \to \infty$. 

The following proposition gives an alternative way to represent  the cumulants. In this new representation, the cumulants are written 
as integrals involving a cyclic product with the analytic Ginibre kernel $K_n$. 

For non-negative integers $\alpha, \beta$ and $n$, we define the differential operators
$$
\mathcal{D}_{\alpha, \beta, n} =\sum_{j=0}^{\min(\alpha, \beta)} {\min(\alpha, \beta) \choose j} (\max(\alpha, \beta))_j n^j \bar \partial^{\alpha-j} \partial^{\beta-j},
$$
where $(r)_j= r (r-1) \cdots (r-j+1)$. Assuming $\beta \geq \alpha$, this can be written in a slightly more compact form as 
$$
\mathcal{D}_{\alpha, \beta, n} = \alpha! n^{\alpha} \partial^{\beta-\alpha} L^{\beta- \alpha}_{\alpha}(-n^{-1} \Delta),
$$
where 
$$
L^k_r:=\sum_{j=0}^r (-1)^j {r+k \choose r-j} \frac{1}{j!} x^j
$$ 
denotes the associated Laguerre polynomial of index $k$ and degree $r$. Similarly, when $\alpha \geq \beta$, we have
$$
\mathcal{D}_{\alpha, \beta, n} = \beta! n^{\beta} \bar \partial^{\alpha-\beta} L^{\alpha-\beta}_{\beta}(-n^{-1} \Delta),
$$
In particular, in the special case $\alpha= \beta=r$ 
we have
\begin{equation}\label{eq:laguerrediffop}
\mathcal{D}_{\alpha, \beta, n}= n^r r! L^0_r(-n^{-1} \Delta).
\end{equation}
\begin{prop} \label{prop:crossterms}
Let $F:\mathbb{C}^k \to \mathbb{C}$ be bounded and smooth. We have
\begin{align*}
\int_{\mathbb{C)}^k} F(z_1, \dots, z_k) [T_n]_{z_1}^{i_1} \overline{[T_n]_{z_2}^{i_1}}K(z_1, z_2)
\cdots [T_n]_{z_k}^{i_k} \overline{[T_n]_{z_1}^{i_k}}K(z_k, z_1) d \mu^k(z_1, \ldots, z_k) \\
= \int_{\C^k} \mathcal{D}_{i_1, \ldots, i_k, n} F(z_1, \ldots, z_k) K_n(z_1, z_2) \cdots
K_n(z_k, z_1) d \mu^k(z_1, \ldots, z_k). 
\end{align*}
where 
$$ \mathcal{D}_{i_1, \ldots, i_k, n}= [\mathcal{D}_{i_k, i_1,n} \big]_{z_1} 
 [ \mathcal{D}_{i_1, i_2,n} ]_{z_2} \cdots [ \mathcal{D}_{i_{k-1}, i_k,n} ]_{z_k}.
$$
\end{prop}
\begin{proof}
We will use that 
\begin{equation} \label{eq:partint}
\int f(z) T_n g(z) d \mu_n(z)= \int \partial f(z) g(z) d \mu_n(z)
\end{equation}
where $f, g \in C^1$. In our applications, $f$ and $g$ will have at most polynomial growth in $z$, so boundary terms do not appear in the partial integration.  We will also need the following: for $j \leq r$, 
\begin{equation} \label{eq:cancel}
\bar \partial^j T_n^r f =(r)_j n^j T_n^{r-j}f
\end{equation}
for an analytic function $f$.  This can be seen as follows: 
\begin{align*} 
\bar \partial^j T_{n}^r f &= (-1)^r \bar \partial^j \sum_{k=0}^r {r \choose k} f^{(k)}(z) (-n \bar z)^{r-k} \\
&=(-1)^{r+j} n^j \sum_{k=0}^{r-j} {r \choose k} (r-k)_j(-n)^{r-k-j} \bar{z}^{r-j-k}f^{(k)}(z) \\
=&(-1)^{r+j} n^j(r)_j \sum_{k=0}^{r-j} {r-j \choose k} (-n \bar z)^{r-j-k} f^{(k)}(z)=n^j (r)_j T_{n}^{r-j} f.
\end{align*}

We integrate by parts in the variable $z_1$.  Suppose that $i_1 \leq i_k$. Using \eqref{eq:partint} and \eqref{eq:cancel}, we obtain
\begin{align*}
& \int_{\C^k} F(z_1, \ldots, z_k) [T_n]_{z_1}^{i_1} \overline{[T_n]_{z_2}^{i_1}}K_n(z_1, z_2)
\cdots [T_n]_{z_k}^{i_k} \overline{[T_n]_{z_1}^{i_k}}K_n(z_k, z_1) d \mu^k \\
&=\sum_{j=0}^{i_1} {i_1 \choose j} n^j (i_k)_j  \partial^{i_1-j}_{z_1} F(z_1, \ldots, z_k)  \overline{[T_n]_{z_2}^{i_1}}K(z_1, z_2)
\cdots [T_n]_{z_k}^{i_k} \overline{[T_n]_{z_1}^{i_k-j}}K_n(z_k, z_1) d \mu^k \\
&= \sum_{j=0}^{i_1} {i_1 \choose j} n^j (i_k)_j \bar \partial^{i_k-j}_{z_1} \partial^{i_1-j}_{z_1}  F(z_1, \ldots, z_k)  \overline{[T_n]_{z_2}^{i_1}}K_n(z_1, z_2)
\cdots [T_n]_{z_k}^{i_k} K_n(z_k, z_1) d \mu^k.
\end{align*}
The last equality holds because $\overline{[T_n]_{z_2}^{i_i} }K_n(z_1,z_2)$ is analytic in $z_1$. If $i_1 \geq i_k$ instead, the  operator acting on $F$ 
is $ \sum_{j=0}^{i_k} {i_k \choose j} n^j (i_1)_j \bar \partial^{i_k-j}_{z_1} \partial^{i_1-j}_{z_1}$. In any case, the operator acting on $F$ is
$[\mathcal{D}_{i_k, i_1,n}]_{z_1}$. 

The claim now follows by performing the same procedure in the other variables. 
\end{proof}
Taking $F= G_k$, it follows directly from this proposition and \eqref{eq:split-kernel} that
\begin{align} \label{eq:fullcumu}
&C_k(\trace_{n,q}g) \\
&= \sum_{0 \leq i_1, \ldots, i_k \leq q-1} \frac{n^{-i_1-\ldots -i_k}}{i_1! \cdots i_k!}
\int_{\C^k} \mathcal{D}_{i_1, \ldots, i_k, n} G_k(z_1, \ldots, z_k) K_n(z_1, z_2) \cdots
K_n(z_k, z_1) d \mu^k(z_1, \ldots, z_k). 
\end{align}

In the pure polyanalytic case, we have the following, slightly simpler version. Using
\eqref{eq:laguerrediffop}, we obtain
\begin{align} \label{eq:purecumu}
&C_k(\trace_{\delta; n,q}g)\\ 
&= \int_{\C^k} \bigg( L_{q-1}(-n^{-1}\Delta_{z_1})\cdots L_{q-1}(-n^{-1}\Delta_{z_k}) G_k(z_1,\ldots,z_k) \bigg) K_n(z_1,z_2)\cdots K_n(z_k,z_1)d\mu^k_n(z_1, \ldots, z_k).
\nonumber
\end{align}

\section{Estimation of the cumulants} \label{sec:cumuest}
In the previous section, we saw that the cumulants $C_k(\trace_{q,n}g)$ and
$C_k(\trace_{\delta;q,n})$  can both be written as a finite series 
\begin{equation} \label{eq:cumulantseries}
\sum_j n^{-j} \int_{\C^k} H_j(z_1, \ldots, z_k) K_n(z_1, z_2)K_n(z_2,z_3) \cdots K_n(z_k, z_1)
d \mu_n^k(z_1, \ldots, z_k).
\end{equation}
where he functions $H_j$ are of the form
$$
H_j(z_1, \ldots, z_k)= \sum_{\alpha} \prod_{m=1}^k f_{j, \alpha, m}(z_m)
$$ 
Here, the sum over $\alpha$ is finite and the functions $f_{j, \alpha, m}$ are bounded and 
continuous and each of them is either compactly supported or identically equal to $1$. In addition, we know that for each $\alpha$ and $j$, the former is the case for at least one $f_{j, \alpha,m}$. 

In this section, we will show that for every $j \geq 2$, the corresponding 
term in the sum \eqref{eq:cumulantseries} tends to zero as $n \to \infty$. Furthermore, we will show that this is the case
also for the $j=1$ term if $H_j$ vanishes on the diagonal, i.e. if $H_j(z, \ldots, z)=0$. These
results are then applied in Section \ref{sec:proof} to prove Theorem \ref{maintheorem}. 

We will start by applying well-known estimates of the kernel Ginibre kernel $K_n$. In fact, essentially the same estimates are also valid for kernels which are defined with more general than Gaussian weights. 
We have the following off-diagonal decay estimate for $K_n$ (Theorem 8.1 in \cite{ameur2010berezin}):
\begin{equation}\label{eq:off-diag}
\lvert K_n(z,w)\rvert^2 e^{-n\left(\lvert z\rvert^2+\lvert w\rvert^2\right)}\leq Cn^2 e^{-c\sqrt{n}\min\{\lvert z-w\rvert, d(z,\partial\D)\}}, \quad z\in \D,w\in\C
\end{equation}
where the constants $C$ and $c$ are absolute constants, in particular independent of $n$. 
We will also need (Proposition 3.6 and p. 1541 in \cite{ameur2010berezin})
\begin{equation} \label{eq:standest}
\lvert K_n(z,w)\rvert^2e^{-n\lvert z\rvert^2-n\lvert w\rvert^2}\leq Cn^2e^{-n(Q(z)-\hat{Q}(z))-n(Q(w)-\hat{Q}(w))},
\end{equation}
where $Q(z)= |z|^2$ and 
$$
\hat{Q}(z)= 
\begin{cases}
|z|^2,  & |z| \leq 1 \\
\log |z|^2+1, & |z| \geq 1.
\end{cases} 
$$
According to this estimate, $\lvert K_n(z,w)\rvert^2e^{-n\lvert z\rvert^2-n\lvert w\rvert^2}$ decays quickly to zero as $n \to \infty$ if either $z$ or $w$ is outside the unit disk.

Let 
$$
\epsilon_n=\epsilon_{n,k}=M_k\frac{\log n}{\sqrt{n}},
$$ 
where $M_k$ is some large constant to be specified later.
We split $\C^k$ into three different regions $\Lambda_n, \Gamma_n$ and $\Omega_n$ as follows:
\begin{align*}
\Lambda_n &=\{z_1\in (1-\epsilon_n)\D, \max_{2\leq j\leq k}\lvert z_j-z_1\rvert\leq \epsilon_n/2\}, \\
\Gamma_n &=\{d(z_1, \partial\D)\leq \epsilon_n, (z_2,\ldots,z_k)\in\C^{k-1}\}, \\
\Omega_n &=\Omega_n^1\cup\Omega_n^2=\{z_1\in(1-\epsilon_n)\D, \lvert z_j-z_1\rvert\geq \epsilon_n/2 \;\text{for some}\; j\}\cup\{z_1\in \C\setminus\D, d(z_1,\partial\D)\geq\epsilon_n\}.
\end{align*}
The following estimate on the domain $\Omega_n$ is a straigthforward application of 
the kernel estimates \eqref{eq:standest} and \eqref{eq:off-diag}. We will write
$$
R_{k,n}(z_1, \ldots, z_k):= K_n(z_1, z_2) \cdots K_n(z_k, z_1).
$$
\begin{prop}
For any $F\in L^{\infty}(\C^k)$, we have that
$$
\int_{\Omega_n}F(z_1,\ldots,z_k)R_{k,n}(z_1,\ldots,z_k)d\mu_n^k(z_1,\ldots,z_k)=O(n^{-1}),
$$
whenever $M_k\geq 2(k+1)/c$, where $c$ is the constant in \eqref{eq:off-diag}.
\end{prop}
\begin{proof}
Defining $f(x)= x - \log x -1$, we can use \eqref{eq:standest} to write
\begin{align*}
\bigg| \int_{\Omega_n^2}F(z_1,\ldots,z_k)R_{k,n}(z_1,\ldots,z_k)d\mu_n^k(z_1,\ldots,z_k)\bigg|  \\
\leq C^{k/2}n^k \Vert F \Vert_{\infty} \int_{|z|^2 \geq (1+\epsilon_n)^2} e^{-n f(|z|^2)} dA(z) 
\bigg[ \int_{\C} e^{-n(Q(z)- \hat{Q}(z))} dA(z) \bigg]^{k-1}
\end{align*}
For $x\geq (1+\epsilon_n)^2$, we estimate by convexity 
$$
f(x) \geq f((1+\epsilon_n)^2)+f^{\prime}((1+\epsilon_n)^2)(x- (1+ \epsilon_n)^2) \geq \epsilon_n^2 + \frac{2 \epsilon_n + \epsilon_n^2}{(1+\epsilon_n)^2}[x-(1+\epsilon_n)^2],
$$
so 
$$
\int_{|z|^2 \geq (1+\epsilon_n)^2} e^{-n f(|z|^2)} dA(z) \leq 
\mathrm{O}(e^{-n\epsilon_n^2})= \mathrm{O}(n^{-M_k^2 \log n}).
$$
Because 
$$
\int_{\C} e^{-n(Q(z)- \hat{Q}(z))} dA(z)= \mathrm{O}(1),
$$
we obtain
$$
\int_{\Omega_n^1}F(z_1,\ldots,z_k)R_{k,n}(z_1,\ldots,z_k)d\mu_n^k(z_1,\ldots,z_k) 
= \mathrm{O}(n^{-\kappa})
$$
 for any desired $\kappa >0$. 
 
Proceeding to the bulk term, we note that at each point $(z_1,\ldots,z_k)\in\Omega_n^1$, there is some index $j$ for which $\lvert z_j-z_{j+1} \rvert\geq \frac{\epsilon_n}{2k}$. 
Choosing $j_0 \geq 1$ to be the smallest such index, we observe that $d(z_{j_0}, \partial \mathbb{D}) \geq \epsilon_n/2$. Now we use the off-diagonal estimate \eqref{eq:off-diag}, and obtain
$$
\lvert K_n(z_{j_0},z_{j_0+1})\rvert e^{-n(\lvert z_{j_0}\rvert^2+\lvert z_{j_0+1}\rvert^2)/2} \leq Cn e^{-\frac{1}{4k}cM_k\log n}=Cn^{1-\frac{1}{4k}cM_k}.
$$
We thus get
\begin{equation} \label{eq:est1}
\lvert R_{k,n}(z_1,\ldots,z_k)e^{-n(\lvert z_1\rvert^2+\ldots +\lvert z_k\rvert^2)} \rvert \leq Cn^{k-\frac{1}{4k}cM_k}.
\end{equation} 
Because 
$$\int_{\C \setminus 2\mathbb{D}}e^{-n(\lvert z \rvert^2- \log \lvert z \rvert^2 -1)} = \mathrm{O}(e^{-an})
$$
for some $a > 0$, we see by using \eqref{eq:standest} that 
\begin{align*}
&\int_{\Omega_n^1}F(z_1,\ldots,z_k)R_{k,n}(z_1,\ldots,z_k)d\mu_n^k(z_1,\ldots,z_k) \\
&= \int_{\Omega_n^1 \cap (2 \mathbb{D})^k}F(z_1,\ldots,z_k)R_{k,n}(z_1,\ldots,z_k)d\mu_n^k(z_1,\ldots,z_k) + \mathrm{O}(e^{-bn})
\end{align*}
for some $b>0$. Applying this with \eqref{eq:est1} gives
$$
\int_{\Omega_n^1}F(z_1,\ldots,z_k)R_{k,n}(z_1,\ldots,z_k)d\mu_n^k(z_1,\ldots,z_k)=O\left(n^{k-\frac{cM_k}{4k}}\right)=O(n^{-1}),
$$
where the last equality follows by the choice $M_k\geq 4k(k+1)/c$.
\end{proof}

The terms $j \geq 2$ in \eqref{eq:cumulantseries} are negligible because of the following proposition. 
\begin{prop}\label{prop:projection-est}
Assume that $f_1\in L^1(\C)$, and $f_2,\ldots f_k \in L^{\infty}(\C)$. Then
$$
\left\lvert \int_{\C^k} \prod_{j=1}^n f_j(z_j)R_{k,n}(z_1,\ldots ,z_k)d\mu_n^k(z_1,\ldots,z_k)\right\rvert \leq n \lVert f_1\rVert_{L^1(\C)}\prod_{j=2}^k\lVert f_j\rVert_\infty.
$$
\end{prop}
\begin{proof}
We write the integral on the left hand side as
$$
\int_{\C^k}\prod_{j=1}^n f_j(z_j)R_{k,n}(z_1,\ldots ,z_k)d\mu_n^k(z_1,\ldots,z_k)=\int_{\C}f(z_1)F(z_1)dA(z_1),
$$
where
\begin{align*}
F(z_1)&=\int_{\mathbb{C}^{k-1}}\prod_{j=2}^{k}f_j(z_j)K_n(z_1,z_2)\cdots K_n(z_k,z_1)d\mu_n^{k-1}(z_2,\ldots ,z_k) e^{-n\lvert z_1\rvert^2}
\\ &=P_n[f_2P_n[\ldots P_n[f_k K_{n,z_1}]\ldots]](z_1)e^{-n\lvert z_1\rvert^2}.
\end{align*}
Introducing the function
$$
F(z_1,z)=P_n[f_2P_n[\ldots P_n[f_k K_{n,z_1}]\ldots]](z)e^{-n\lvert z\rvert^2},
$$
we may write $F(z_1)=F(z_1,z_1)$. Now,
$$
\left\lvert F(z_1,z)\right\rvert \leq \sqrt{K_n(z,z)}\lVert F(z_1,\cdot)\rVert_{L^2(d\mu_n)}e^{-n\lvert z\rvert^2}\leq \sqrt{K_n(z,z)}\sqrt{K_n(z_1,z_1)}\left[\prod_{j=2}^k \lVert f_j\rVert_\infty\right] e^{-n\lvert z\rvert^2}
$$
where the first inequality is by Cauchy-Schwarz, and the second follows since projections have norm one.
It follows that
$$
\left\lvert\int_\C f_1(z_1)F(z_1)dA(z_1)\right\rvert\leq \prod_{j=2}^k\lVert f_j\rVert_\infty \int_\C\lvert f_1(z_1)\rvert K_n(z_1,z_1)e^{-n\lvert z_1\rvert^2}dA(z_1),
$$
which by the trivial estimate $K_n(z,z)e^{-n\lvert z\rvert^2}\leq n$ proves the assertion.
\end{proof}
The following lemma deals with the term corresponding to $j=1$ in \eqref{eq:cumulantseries}.
\begin{lemma}\label{lem:diagonal}
Let $F(z_1,\ldots, z_k)$ function of the form
$$
F(z_1,\ldots, z_k)=\sum_{\alpha=1}^M\prod_{j\leq k}f_{\alpha,j}(z_j), \quad M \geq 1
$$
where $f_{\alpha, j} \in C^2(\C) \cap L^{\infty}(\C)$. Assume furthermore that
$$
F(z,z,\ldots,z)=0,\quad z\in\C.
$$
Then
$$
n^{-1}\int_{\C^k}F(z_1,\ldots,z_k)R_{k,n}(z_1,\ldots,z_k)d\mu_n^k(z_1,\ldots,z_k)=o(1).
$$
\end{lemma}
\begin{proof} 
We split the integral over $\C^k$ as
\begin{align*}
\int_{\C^k}F R_{n,k}d\mu_n^k=\int_{\Lambda_n}F R_{n,k}d\mu_n^k+\int_{\Gamma_n}F R_{n,k}d\mu_n^k+O(n^{-1}).
\end{align*}
Turning to the integral over $\Gamma_n$, we may assume without loss of generality that $F(z_1,\ldots,z_k)=f_1(z_1)\cdots f_k(z_k)$ for some bounded continuous functions $f_j$, since vanishing of $F$ on the diagonal will not be used in this case. 
It follows immediately from Proposition~\ref{prop:projection-est} that
\begin{align*}
\int_{\Gamma_n}F R_{k,n}d\mu_n 
&=\int_{\C^k}f_1(z_1) \,1_{ \{ \mathrm{dist}(z_1, \partial \mathbb{D}) \leq M_k \log n / \sqrt{n} \}}(z_1)  \prod_{j=2}^kf_j(z_j)R_{k,n}(z_1, \ldots, z_k) d\mu_n^k(z_1, \ldots, z_k) \\
&\leq O \big(n\lVert f_1 1_{ \{ \mathrm{dist}(z, \partial \mathbb{D}) \leq M_k \log n / \sqrt{n} \}} \rVert_{L^1} \big) =O(\sqrt{n}\log n).
\end{align*}
For the integral over $\Lambda_n$, we estimate $F$ near $\Lambda_n$ by a Taylor expansion, and obtain
$$
\delta_n:=\lVert F_{\vert\Lambda_n}\rVert_\infty =(k-1)\max_{2\leq j\leq k}( \lVert \partial_{z_j} F \rVert_{L^\infty( \Lambda_n)} + \lVert \bar  \partial_{z_j} F \rVert_{L^\infty(\Lambda_n)})  \frac{\epsilon_n}{2}+O(\epsilon_n^2)=O(\epsilon_n).
$$
The crude estimate $\lvert R_{k,n} \rvert e^{-n|z_1|^2- \ldots- n|z_k|^2} \leq n^k$ and $Vol(\Lambda_n)=O(\epsilon_n^{2(k-1)})$ then gives
$$
\int_{\Lambda_n}\lvert F\rvert R_{k,n}d\mu_n\leq \int_{\Lambda_n}C\delta_n n^k=O\left(\delta_n n (M_k\log n)^{2k-2}\right)=O\left(n\frac{(M_k\log n)^{2k-1}}{\sqrt{n}}\right)=o(n),
$$
which completes the proof.
\end{proof}

\section{Proof of the main theorem} \label{sec:proof} 
\begin{proof}[Proof of Theorem \ref{maintheorem}]
We will start by proving \eqref{eq:purevariance}. By \eqref{eq:purecumu}, we know that 
$C_k(\trace_{\delta; n,q})$ can be written as a finite series of the form \eqref{eq:cumulantseries}. 
By Proposition \ref{prop:projection-est}, it is enough pick the terms corresponding 
to the powers $n^0$ and $n^{-1}$ from the expression
$$
L^0_{q-1}(n^{-1}\Delta_{z_1})\cdots L^0_{q-1}(n^{-1}\Delta_{z_k}) G_k(z_1,\ldots,z_k).
$$
We get
\begin{align} \label{eq:cumulant_mainterms}
&C_k(\trace_{\delta, n,q}) \\
&=\int_{\C^k}\left(I+n^{-1}(q-1)\sum_{j=1}^k\Delta_{z_j}\right)G_k(z_1,\ldots,z_k)R_{k,n}(z_1,\ldots,z_k)d\mu_n^k(z_1,\ldots,z_k)+O(n^{-1}). \nonumber
\end{align}
Suppose first that $k \geq 3$. The term corresponding to the identity operator $I$ in this formula is just the cumulant of order $k$ in the analytic case $q=1$. By the result of Rider and Vir\'{a}g \cite{rider2007noise}, 
this expression vanishes as $n \to \infty$. For the $n^{-1}$ term, it was shown in Lemma 3.3 of \cite{ameur2011fluctuations} that $\sum_{j=1}^k \Delta_{z_j} G_k(z_1, \ldots, z_k)$ vanishes on the diagonal $z_1 = \cdots = z_k$.  
The claim now follows by Lemma \ref{lem:diagonal}. 
 
It remains to calculate the second cumulant, which is equal to the variance. We use 
\eqref{eq:cumulant_mainterms} again. Recall that 
$G_2(z,w)= \frac12 (f(z)-f(w))^2$. By the result of \cite{rider2007noise}, we have 
\begin{equation} \label{eq:ridvir-variance}
\int_{\C^k}G_2(z_1, z_2)R_{2,n}(z_1,z_2)d\mu_n^k(z_1,z_2) \to 
\lVert \bar\partial g\rVert_{L^2(\D)}^2+\frac12 \lVert g\rVert_{H^{1/2}}^2.
\end{equation}
To analyse the terms involving the Laplacians in \eqref{eq:cumulant_mainterms}, we calculate
$$
(\Delta_z+\Delta_w)(g(z)-g(w))^2=2(\Delta g(z)-\Delta g(w))(g(z)-g(w))+2\lvert \bar\partial g(z)\rvert^2+2\lvert \bar\partial g(w)\rvert^2.
$$
We now observe that  $2(\Delta g(z)-\Delta g(w))(g(z)-g(w))$ vanishes when $z=w$. With Lemma \ref{lem:diagonal} we now have
\begin{align*} 
&C_k(\trace_{\delta, n,q})\\
&= \lVert \bar\partial g\rVert_{L^2(\D)}^2+\frac12 \lVert g\rVert_{H^{1/2}}^2+ \frac{1}{2n} (q-1) \int_{\C^2}
(2\lvert \bar\partial g(z)\rvert^2+2\lvert \bar\partial g(w)\rvert^2 \lvert K_n(z,w) \lvert^2 d \mu^2(z_1, z_2) + \mathrm{o}(1) \\
&= \lVert \bar\partial g\rVert_{L^2(\D)}^2+\frac12 \lVert g\rVert_{H^{1/2}}^2 + 
\frac{2(q-1)}{n} \int_{\C^2} \lvert \bar \partial g(z)\lvert^2 \lvert K_n(z_1, z_2) \lvert^2 d \mu_n^2(z_1, z_2) + \mathrm{o}(1) \\
&=\lVert \bar\partial g\rVert_{L^2(\D)}^2+\frac12 \lVert g\rVert_{H^{1/2}}^2 + 
\frac{2(q-1)}{n} \int_{\C} \lvert \bar \partial g(z)\lvert^2 K_n(z,z) d \mu_n(z) + \mathrm{o}(1) \\
&\to (2q-1) \int_{\mathbb{D}} \lvert \bar \partial g \lvert^2 dA+ \frac12 \lVert g\rVert_{H^{1/2}}^2.
\end{align*}
To obtain the last limit, we used that $\frac1n K_n(z,z) d\mu_n(z) \to 1_{\mathbb{D}}$ weakly. 

We will now turn to the proof of \eqref{eq:fullvariance}. We have 
\begin{align*}
&C_k(\trace_{q, n})\\
&= \sum_{1 \leq i_1, \cdots, i_k \leq q-1} \frac{n^{-i_1 - \ldots - i_k}}{i_1! \cdots i_k} \int_{\C^k} G_k(z_1, \dots, z_k) [T_{n,1}]_{z_1}^{i_1} \overline{[T_{n,1}]_{z_2}^{i_1}}K_n(z_1, z_2)
\cdots [T_{n,1}]_{z_k}^{i_k} \overline{[T_{n,1}]_{z_1}^{i_k}}K_n(z_k, z_1) d \mu^k 
\end{align*}
As with pure polyanalytic kernels, we assume first that $k \geq 3$. Let us now fix a multi-index $(i_1, \ldots, i_k)$ and estimate the corresponding term in the previous expression. By 
Proposition \ref{prop:crossterms}, we get
\begin{align*}
\frac{n^{-i_1 - \ldots - i_k}}{i_1! \cdots i_k} \int_{\C^k} G_k(z_1, \dots, z_k) [T_{n,1}]_{z_1}^{i_1} \overline{[T_{n,1}]_{z_2}^{i_1}}K_n(z_1, z_2)
\cdots [T_{n,1}]_{z_k}^{i_k} \overline{[T_{n,1}]_{z_1}^{i_k}}K(z_k, z_1) d \mu^k \\
= \frac{n^{-i_1 - \ldots - i_k}}{i_1! \cdots i_k} \int_{\C^k} \mathcal{D}_{i_1, \ldots, i_k, n} G_k(z_1, \ldots, z_k) K_n(z_1, z_2) \cdots
K_n(z_k, z_1) d \mu^k(z_1, \ldots, z_k) 
\end{align*}
We have 
$$
 n^{-i_1 - \ldots - i_k} \mathcal{D}_{i_1, \ldots, i_k, n} G_k(z_1, \ldots, z_k) = 
 \mathrm{O}(n^{I(i_1, \ldots, i_k)}),
 $$
where
$$
I(i_1, \ldots, i_k) := -i_1- \ldots - i_k+ \mathrm{min}(i_1, i_k) + \mathrm{min}(i_1, i_2)+ \ldots +\mathrm{min}(i_{k-1}, i_k).
$$
Clearly, $I(i_1, \ldots, i_k) \leq 0$. We have that $I(i_1, \ldots, i_k)=0$ if and only if $i_1= \ldots= i_k$, and this case is treated above. The terms where $I(i_1, \ldots, i_k) \leq -2$ are negligible by proposition \ref{prop:projection-est}, so the only case that needs further analysis is 
$I(i_1, \ldots, i_k)=-1$. This happens exactly when there exists a unique $l \in \{ 1, 2, \ldots, k \}$ such that $i_l< i_{l+1}$ and for this $l$ it holds $i_l = i_{l+1}-1$. Here and later, we identify $i_{k+j}$ with $i_j$. This condition on the multi-index $(i_1, \ldots, i_k)$ can be rephrased as follows:  $(i_1, \ldots, i_k)= \mathbf{d}(j,r, i)$ for some $1 \leq j \leq k$, 
$0 \leq r \leq k-2$ and $1 \leq i \leq q-1$, where $\mathbf{d}(j,r, i)$ is defined to be the multi-index $(i_1, \ldots, i_k)$ 
which satisfies $i_m=i$ for $m \in \{ j, \ldots, j+r \}$
and $i_m = i-1$ for $i_m \in \{i_1, \ldots, i_k \} \setminus \{i_j, \ldots, i_{j+r} \}$. 
If  $(i_1, \ldots, i_k)= \mathbf{d}(j,r, i)$, we have
\begin{align} \label{eq:maincontr}
 &\frac{n^{-i_1-\ldots-i_k} }{i_1! \cdots i_k!} \mathcal{D}_{i_1, \ldots, i_k, n} G_k(z_1, \ldots, z_k) \\
 &= \frac{n^{-1} i_1! \cdots i_{j+r}!(i_{j+r+1}+1)! \cdots i_k!}{i_1! \cdots i_k}  \partial_{z_{j}} \bar \partial_{z_{j+r+1}} G_k(z_1, \ldots, z_k) + \mathrm{O}(n^{-2})  \nonumber \\
 &= n^{-1}i \partial_{z_{j}} \bar \partial_{z_{j+r+1}} G_k(z_1, \ldots, z_k) + \mathrm{O}(n^{-2}), \nonumber
\end{align}
where we collected from each operator $\mathcal{D}_{i_j, i_{j+1}, n}$ the term which has the highest degree in $n$. The $\mathrm{O}(n^{-2})$ are negligible by Proposition \ref{prop:projection-est}.
Summing up the contributions from all such $(i_1, \ldots, i_k)$, we can write the $k$'th cumulant as 
\begin{align*}
&\sum_{j=1}^{q}\int_{\C^k} G_k(z_1, \ldots, z_k) K_{\delta; j,n}(z_1, z_2) \cdots K_{\delta; j,n}(z_k,z_1)d \mu^k_n(z_1, \ldots, z_k) \\
&+ \int_{\C^k} \big[\sum_{j=1}^{k} \sum_{r=0}^{k-2} \sum_{i=1}^{q-1} i \partial_j \bar \partial_{j+r+1} \big] G_k(z_1, \ldots, z_k) R_{k,n}(z_1, \ldots, z_k) d \mu^k_n(z_1, \ldots, z_k)
+ \mathrm{O}(n^{-1}) \\
=  &\sum_{j=1}^{q}\int_{\C^k} G_k(z_1, \ldots, z_k) K_{\delta; j,n}(z_1, z_2) \cdots K_{\delta; j,n}(z_k,z_1)d \mu^k_n(z_1, \ldots, z_k) \\
&+\frac{(q-1)(q-2)}{2} \int_{\C^k}  \sum_{j \neq l} \partial_j \bar \partial_l G_k(z_1, \ldots, z_k) \big] R_{k,n}(z_1, \ldots, z_k) d \mu^k_n(z_1, \ldots, z_k) + \mathrm{O}(n^{-1})
\end{align*}
By Lemma 3.4 in \cite{ameur2011fluctuations}, $\sum_{l \neq j} \partial_j \bar \partial_l G_k(z_1, \ldots, z_k)$ vanishes on $z_1= \ldots = z_k$. This together with lemma \ref{lem:diagonal} and the result in pure polyanalytic case proves that the cumulants of orders $k \geq 3$ vanish in the limit. 

Next, we calculate the limiting variance. 
Suppose $0 \leq i_1 \leq i_2 \leq q-1$.  Again, we only need to analyse the case 
$I(i_1, i_2)= -1$, and this happens exactly when $i_2= i_1+1$. Now, using \eqref{eq:maincontr} again,
\begin{align*}
&\frac{n^{-2i_1-1 }}{i_1! (i_1+1)!} \int_{\mathbb{C}^2} G_2(z_1, z_2) [T_n]_{z_1}^{i_1} \overline{[T_n]_{z_2}^{i_1}}K(z_1, z_2) 
[T_n]_{z_2}^{i_1+1} \overline{[T_n]_{z_1}^{i_1+1}}K(z_2, z_1) d \mu^2(z_1, z_2) \\
&= n^{-1} (i_1+1) \int_{\mathbb{C}^2} \bar \partial_{z_1}  \partial_{z_2} G_2(z_1, z_2) 
K_n(z_1, z_2) K_n(z_2, z_1) d \mu^2(z_1, z_2)   + \mathrm{O}(n^{-1}) \\
&=- n^{-1}(i_1+1) \int_{\C^2} \bar \partial_{z_1}g(z_1)\partial_{z_2}g(z_2)|K_n(z_1, z_2)|^2 d \mu^2(z_1,z_2) + \mathrm{O}(n^{-1}) \\ 
&= - n^{-1}(i_1+1) \int_{\C^2} \bigg[\bar \partial_{z_1}g(z_1)(\partial_{z_2}g(z_2)- \partial_{z_1}g(z_1)) + |\bar \partial_{z_1} g(z_1)|^2 \bigg] K_n(z_1, z_2)|^2 d \mu^2(z_1,z_2) + \mathrm{O}(n^{-1}) \\ 
&= -(i_1+1) \int_{\mathbb{D}} |\bar \partial g(z)|^2 K_n(z,z)d\mu_n(z) + \mathrm{o}(1) 
\to -(i_1+1) \int_{\mathbb{D}} |\bar \partial g(z)|^2 dA(z)
\end{align*}
Here, the second last equality followed from Lemma \ref{lem:diagonal} and integrating out the second variable. For the last limit, we used again the weak convergence
$K_n(z,z)d\mu_n(z) \to 1_{\mathbb{D}}$. Summing up all such contributions gives
\begin{align*}
&\int_{\C^2} G_2(z_1, z_2) |K_{q,n}|^2 d \mu^2(z_1, z_2) \\
&= \int_{\C^2} G_2(z_1, z_2)  \bigg[ \sum_{j=1}^q |K_{\delta; j,n}|^2 + 2 \mathrm{Re} \sum_{1 \leq j < l \leq k} K_{\delta; j,n}(z_1, z_2) K_{\delta; l, n}(z_2, z_1) \bigg] d \mu(z_1, z_2) \\
&\to \sum_{j=1}^q (2j-1) \| \bar \partial g \|_{L^2{D})}^2 + \frac{q}{2} \|g \|^2_{H^{1/2}(\partial \mathbb{D})} - 2\sum_{j=0}^{q-2} (j+1) \| \bar \partial g \|^2= q(\| \bar \partial g \|_{L^2(\D)}^2 + \frac12 \|g \|^2_{H^{1/2}(\partial \mathbb{D})})
\end{align*}
\end{proof}
%
%
%
%
 \bibliography{fluct9}
\bibliographystyle{amsplain}

\end{document}